\documentclass[runningheads]{llncs}

\usepackage{graphicx}
\usepackage{geometry}
\usepackage{amsfonts,amsmath,amssymb,nicefrac}
\usepackage{amsfonts}
\usepackage{bbold}
\usepackage{color}
\title{On overall measure of non-classicality of $N$-level quantum system and its universality  in the  large $N$ limit}

\author{V. Abgaryan\inst{1,2}\orcidID{0000-0001-6713-4095} \and
A. Khvedelidze\inst{1,3,4}\orcidID{0000-0002-5953-0140} \and
I. Rogojin \inst{1}\orcidID{0000-0001-6440-5451
}}
\authorrunning{V. Abgaryan et al.}

\institute{Laboratory of Information Technologies, Joint Institute for Nuclear Research, Dubna, Russia\and
Peoples’ Friendship University of Russia (RUDN University)
6 Miklukho-Maklaya St, Moscow, 117198, Russian Federation \and
A. Razmadze Mathematical Institute, Iv. Javakhishvili Tbilisi State University, Tbilisi, Georgia\and
Institute of Quantum Physics and Engineering Technologies, Georgian Technical University, Tbilisi, Georgia\\
\email{vahagnab@googlemail.com}}

\begin{document}

\maketitle
\begin{abstract}In this report we are aiming at introducing a global measure of non-classicality of the state space of $N$-level quantum systems and estimating it in the limit of large $N$. For this purpose we employ the Wigner function negativity as a non-classicality criteria. Thus, the specific volume of the support of negative values of Wigner function is treated as a measure of non-classicality of an individual state. Assuming that the states of an $N$-level quantum system are distributed by Hilbert-Schmidt measure (Hilbert-Schmidt ensemble), we define the global measure as the average non-classicality of the individual states over the Hilbert-Schmidt ensemble.
We present the numerical estimate of this quantity as a result of random generation of states, and prove a proposition claiming its exact value in the limit of  $N\to \infty$.


\end{abstract}

\section{Introduction}
With the rise of quantum information and computation paradigms alongside the adjacent fields, one time and again encounters the characterization "{\it non-classical}" when describing the quantum states involved.
It must be pointed out that the notion of non-classicality of quantum states is not well defined. Under this label, we usually understand  the effects predicted by quantum mechanics which are incomprehensible from the standpoint of classical intuition. These include everything spanning from quantum entanglement and other purely quantum correlations to sub-poissonian statistics and squeezing of electromagnetic fields. Quite often these become resources for new powerful techniques, as is the case, for example, with quantum entanglement. A question of quantitative description of the degree of non-classicality and hence of the resource itself arises here. Obviously, due to the wideness and vagueness of the question, it would be naive to assume the existence of a universal measure encompassing the intensities  of all the quantum effects. However, it seems that the central object of quantum mechanics on the phase space, the Wigner function, somehow encodes   the crucial information about the non-classical features of the state through the property of having negative values. Indeed, to name just a few examples: it has been shown that quantum circuits where the initial state together with the quantum operations is representable with positive Wigner functions can be classically efficiently simulated~\cite{Mari}; s-waves are entangled if and only if corresponding Wigner function has negative domains~\cite{Dahl}; the negativity volume of the Wigner function is an entanglement indicator for hybrid qubit\--bosonic states if certain conditions are met~\cite{Ievgen}, etc. 

Elaborating on the property of Wigner function to have negative values, several measures of non-classicality  have been introduced (see~\cite{Zyczkowski} and references therein). Here, we generalize these  well-established ideas from the level of individual states to the whole state space. 

The article is organized as follows. In the next section we introduce the necessary basics about the Wigner function. In section~\ref{section3}  we define the main quantities  which will be used in the rest of paper. Section~\ref{section4.1}  contains  results on the global measure of non-classicality  of density matrices from the Hilbert-Schmidt ensemble. Finally,  in section~\ref{subsection4.1} the  
analysis of behavior of the introduced measure of non-classicality for large $N$ is given. 

\section{The Wigner function of a density operator of $N$-level system}\label{section2} 
\label{sec:WF}

Wigner function~\cite{Wigner} was introduced in an attempt of phase space description of quantum mechanics. For quantum states represented by a density operator $\hat{\rho} \in \mathcal{D}(L^2(\mathbb{R}^n))$ acting on the Hilbert space $\mathcal{H}=L^2(\mathbb{R^n})$ the Wigner function of $\hat{\rho}$ is defined over a phase space $(\mathbb{R}^{2n}\,, w )$ with the standard  symplectic  2-form $w:=\sum_{j}^n\,dp_j\wedge d q_j
$ and is given by the 
so-called Wigner transform:
\begin{equation}
\label{eq:WFC}
   W_{\hat{\rho}}(\boldsymbol{q},\boldsymbol{p})= \left(\frac{1}{2 \pi \hbar}\right)^{n}\int_{\mathbb{R}^n}\,d{\boldsymbol{\eta}}\,\left\langle  \boldsymbol{q}-\frac{\boldsymbol{\eta}}{2}\middle|\,
   \hat{\rho}\,\middle|{\boldsymbol{q}}+\frac{\boldsymbol{\eta}}{2}\right\rangle e^{\displaystyle{\frac{i}{\hbar}\boldsymbol{p}\boldsymbol{\eta}}}.
\end{equation}
According to the Weyl-Wigner formalism one can establish an invertible map between  
the self-adjoint semipositive definite operator $\hat{\rho} $ and its Weyl symbol $(2\pi\hbar)^n W_{\hat{\rho}}(\boldsymbol{q},\boldsymbol{p})\,$ in (\ref{eq:WFC}) 
\begin{equation}
\label{eq:GWWM}
\hat{\varrho}\  \rightleftarrows \ W_{\rho} (\boldsymbol{q},\boldsymbol{p})\,. 
\end{equation}
Generalizing the Weyl-Wigner mapping 
(\ref{eq:GWWM})
to the case of an arbitrary self-adjoint operator,  
\begin{equation}
\label{eq:WWMapping}
  \hat{A} \rightleftarrows W_A(\boldsymbol{q},\boldsymbol{p})\,,
\end{equation}
the quantum mechanical prediction  of  the operator, i.e.,  $\mathbb{E}[\hat{A}]=\mbox{tr}[\hat{\rho} \hat{A}]$ is expressible in the 
form of conventional
ensemble average in classical mechanics defined as the mean of an operator symbol $A(\boldsymbol{q},\boldsymbol{p})$ over  the phase-space with the distribution $W_\varrho : $
\begin{equation}
\mathbb{E}[\hat{A}]= \overline{A}\,
\qquad 
\overline{A}:=\int_{\mathbb{R}^{2n}}\,d\boldsymbol{p}\,d\boldsymbol{q}\, W_A(\boldsymbol{q},\boldsymbol{p})\, W_{\rho}(\boldsymbol{p},\, \boldsymbol{q})\,.
\end{equation}
However, the similarity between the quantum and classical expressions is somewhat illusive. Though, the marginal distributions of momenta on one side and of coordinates on the other are true probability  density functions, due to the limitation of simultaneous measurements of coordinates and momenta in quantum mechanics by Heisenberg uncertainty principle, Wigner function is not free from "faults". Namely, it may be shown that there are states for which Wigner function has negative values. Hence it can't be considered as a true probability density function, and is usually called a quasiprobability density function.

As it was mentioned above the Wigner transform is well adapted to the case of a quantum mechanical system associated to the Hilbert space $\mathcal{H}=L^2(\mathbb{R^n})\,.$ 
The natural question arises how to deal with other quantum systems whose Hilbert space $\mathcal{H}$ is different from 
$L^2(\mathbb{R^n})\, ?$ 
In 1957,  based on the Weyl-Wigner approach, R.L.Stratonovich formulated  ~\cite{Stratonovich}   general principles of constructing the mapping (\ref{eq:WWMapping}),  which should be satisfied for any quantum system associated to some Hilbert space.  These principles, later on, received the name of Stratonovich-Weyl (SW) correspondence.  
Since in the present note we are interested in quantification of ``quantumness'' in systems whose Hilbert space is $\mathcal{H}=\mathbb{C}^N$, below   basics of SW correspondence are reproduced in a form  adapted  to the case of finite-dimensional quantum systems.

The basic idea of realisation of mapping  (\ref{eq:WWMapping}) is to use the kernel operator $\Delta(\Omega_N)$ defined over the symplectic manifold $\Omega_N$ endowed with some symplectic 2-form. The mapping is given by formulae 
\begin{eqnarray}
\label{eq:DMWigner}
&&W_A(\Omega_N) = \mbox{tr}\left(A \Delta(\Omega_N)\right)\,,\\
&&\hat{A} = \int_{\Omega_N} \mathrm{d}\Omega_N\, \Delta(\Omega_N) W_A(\Omega_N) \,;
\end{eqnarray}
 Here the kernel $\Delta(\Omega_N)$  is self-dual, in sense that the same kernel defines as direct as well an inverse mapping (\ref{eq:DMWigner}), and it is the so-called Staratonovich-Weyl kernel.  
According to the Stratonovich-Weyl principles in order to have a correct phase-space formulation of quantum theory SW kernel should provide fulfilment of the following compulsory requirements:
\begin{itemize}
    \item the kernel  must be Hermitian, 
    $\Delta(\Omega_N)^{\dagger}=\Delta(\Omega_N)$ guaranteeing the reality of symbols; 
    \item the kernel  must be the trace class operator, i.e.,  $\int_{\Omega_N}d\Omega_N\, \Delta(\Omega_N)=1, $  ensuring completeness of quantum states as well as classical ones;
    \item  the unitary symmetry  of states $\rho^{\prime}=g\,\rho g^{\dagger}\,, \ g \in SU(N)$ induces the adjoint  transformation of SW kernel,  $\Delta(\Omega_N^{\prime})
    =g^{\dagger} \Delta(\Omega_N) g$ where $\Omega_N^{\prime}$ is an image of point $\Omega_N$ under the action of $g\,.$ 
\end{itemize}
It has been shown~\cite{AbKh1} that for $N$-level quantum system the Stratonovich-Weyl correspondence clauses admit simple formulation in the form of algebraic equations on  spectrum of SW kernel:
\begin{equation}
\label{eq:Master}
    \mbox{tr}[\Delta(\Omega_N)]=1 \quad {\text{and}}\quad \mbox{tr}[\Delta(\Omega_N)^{2}]=N\,.
\end{equation}
These equations leave $N-2$ parametric freedom of choice of the spectrum of SW kernel. 
Taking into account this ambiguity 
we can write the SVD decomposition of SW kernel 
\begin{equation}
\label{eq:SVDSW}
\Delta(\Omega_N|\boldsymbol{\nu}) = U(\Omega_N)\, P(\boldsymbol{\nu})\, U^{\dagger}(\Omega_N)\,,
\end{equation}
where  $P(\boldsymbol{\nu})$ is a diagonal matrix whose elements are 
specifically ordered eigenvalues of the SW kernel,  $\boldsymbol{spec}(\Delta)=\{\pi_1(\boldsymbol{\nu}), \pi_2(\boldsymbol{\nu}), \dots, 
\pi_N(\boldsymbol{\nu}) \}\,$. Eigenvalues $\pi(\boldsymbol{\nu})$ are functions of a real $(N-2)$- tuple $\boldsymbol{\nu}=\left(\nu_{1},\,\cdots,\,\nu_{N-2}\right)$ parameterising the moduli space of solutions to (\ref{eq:Master}). 
Hereafter, dealing with the Wigner function of density matrix $\rho$ we will point at this ambiguity by explicitly writing dependence of SW kernel on the moduli space  parameters  $\boldsymbol{\nu}$:
\begin{equation}
    W_{\rho}^{(\boldsymbol{\nu})}(\Omega_N) = \mbox{tr}[\rho\,\Delta(\Omega_N|\boldsymbol{\nu})]\,.
\end{equation}
See more on the moduli space of parameters in~\cite{AVA1}.

Finally, a few remarks on symplectic space $\Omega_N\,$ are  in order. From the SVD decomposition (\ref{eq:SVDSW}) it follows that its 
 structure, particularly its dimension depends on the choice of kernel. 
Now, assuming that its isotropy group $H\in U(N)$ is of the form 
\[
H_{\boldsymbol{k}}={U(k_1)\times U(k_2) \times U(k_{s+1})}\,,
\]
then the corresponding phase-space $\Omega_N$ can be identified with a complex flag variety  
$
\mathbb{F}^N_{d_1,d_2, \dots, d_s}=
{U(N)} / {H}\,, 
$
where 
$(d_1, d_2, \dots, d_s)$ are  positive integers with sum
$N $, such that  $k_1=d_1$ and $k_{i+1}=d_{i+1}-d_i$ with $d_{s+1}=N\,.$ 
Therefore,  each SW kernel is in one-to one correspondence with a point 
of  moduli space (with $\boldsymbol{\nu}$\-- being the corresponding coordinate)  and it is defined over the phase  
$\Omega_{N,\boldsymbol{k}}$, member of the finite family of flag varieties labeled by an integer $(s+1)$-tuple 
$\boldsymbol{k}=(k_1, \dots , k_{s+1})\,.$
The volume form on $\Omega_{N,\boldsymbol{k}}$ is determined by the 
bi-invariant normalised Haar measure $d\mu_{SU(N)}$ on $SU(N)$  group ~\cite{AbKh1}:
\begin{equation}
    d \Omega_{N,\boldsymbol{k}}=N\,\mbox{Vol}(H_{\boldsymbol{k}})\,\frac{d\mu_{SU(N)}}{d\mu_{H_{\boldsymbol{k}}}}\,,
\end{equation}
where $d\mu_{H_{\boldsymbol{k}}}$ is the bi-invariant measure over the isotropy group $H_{\boldsymbol{k}}$.

\section{Measures of non-classicality of state and overall quantum system}
\label{section3}

Before introducing the main quantity we are interested in, it is worth to remind a few auxiliary notions. We begin with the definition of the state space $\mathfrak{P}_N$ of an  $N$-level quantum system. 

\begin{definition}The state space 
$\mathfrak{P}_N$ is a $N^{2}-1$ dimensional subset in the space of $N \times N$ complex matrices $M_{N}(\mathbb{C})$ given by following conditions:
\begin{equation}\label{eq:StateSpace}
    \mathfrak{P}_N =\{ X \in M_N(\mathbb{C}) \ |\ X=X^\dagger\,,\quad  X \geq 0\,,  \quad \mathrm{tr}\left( X \right) = 1   \}\,.
\end{equation}
\end{definition}
Let, $\Delta(\Omega_N\, |\, \boldsymbol{\nu})\,$ be the Stratonovich-Weyl (SW) kernel with moduli parameter $\boldsymbol{\nu}\,.$ 
Due to possible symmetries of state  $\rho$  and SW kernel the corresponding WF function has domain of definition not over the whole $\Omega_N$, but is restricted  to its certain  subset.  Having in mind this fact  we  introduce two additional definitions.
\begin{definition}
 $\Omega_{N}[\rho\,|\,{\boldsymbol{\nu}}] \in \Omega_N$
 represents a support of WF associated to a given state $\rho\in \mathfrak{P}_N\,$ and SW kernel.
\end{definition}
\begin{definition}
 We call $\Omega_{N}^{(-)}[\rho\,|\,{\boldsymbol{\nu}}]$  the negative support  of the Wigner function associated to a given SW kernel and state $\rho\in \mathfrak{P}_N\,,$
\begin{equation}
\Omega_{N}^{(-)}[\rho\,|\,{\boldsymbol{\nu}}]
=\left\{\omega \in \Omega_{N}[\rho\,|\,{\boldsymbol{\nu}}] \,\, |\,\,  W_{\rho}^{\boldsymbol{\nu}}(\omega) <0 \right\}.
\end{equation}
\end{definition}
Associating non classicality  with the discrepancy  between positivity requirement on classical  probability distribution and a  property of the Wigner function to attain negative values,  we introduce a  measure quantifying quantumness of state via  a relative volume of the subset of phase space where this discrepancy occurs.
The next definitions give  formalization of this idea.
\begin{definition}
For a state $\rho$ of an $N$-dimensional quantum system we define its non-classicality measure (or quantumness) $\mathfrak{Q}_N
[\rho\,, \boldsymbol{\nu}]$ as
\begin{equation}
 \label{eq:NCM}
 \mathfrak{Q}_N[\rho\,, {\boldsymbol{\nu}}]=
 \frac{\mathrm{Vol}(\Omega_{N}^{(-)}[\rho\,|\,{\boldsymbol{\nu}}])}{\mathrm{Vol}
 \left(\Omega_{N}[\rho\,|\,{\boldsymbol{\nu}}]\right)}\,.
\end{equation}
\end{definition}
It is necessary to note that in definition (\ref{eq:NCM}) it is assumed that the volume is evaluated using the symplectic volume form which is a projection of the corresponding volume form on the phase space $\Omega_N$ to the subset $\Omega_{N}[\rho\,|\,{\boldsymbol{\nu}}]$. 
\begin{definition}
We call the following unions, 
\begin{equation}
\Omega_{N}[\boldsymbol{\nu}]
=\bigcup_{\rho \in \mathfrak{P}_N} \Omega_{N}[\rho\,|\,{\boldsymbol{\nu}}]\,, \quad \mathrm{and} \quad \Omega^{(-)}_{N}[\boldsymbol{\nu}]=\bigcup_{\rho \in \mathfrak{P}_N} \Omega_{N}^{(-)}[\rho\,|\,{\boldsymbol{\nu}}]
\end{equation}
as  the ``symplectic superspace'' and the collection of supports of  negativity of the Wigner function
will be called correspondingly as ``negativity supersupport''.
\end{definition}
This definitions are in given in a sense   of the famous Wheeler’s superspace notion in General Relativity
(see \cite{Wheeler}). Basically $\Omega_{N}[\boldsymbol{\nu}]$ is the collection of  the supports of the WF of all possible states of $N$-level quantum system with fixed SW kernel. 
Following  the same logic as  before one can  introduce the  measure of  quantumness on the ``symplectic superspace'' as well.
\begin{definition}
For a given SW kernel,  the global non-classicality measure  $\mathfrak{Q}_N[\boldsymbol{\nu}]$ of $N$-level quantum  system is 
\begin{equation}
\label{eq:GM}
\mathfrak{Q}_{N}[\boldsymbol{\nu}]=
\frac{\mathrm{Vol}_{\mathrm{g}}(\Omega_{N}^{(-)}[\boldsymbol{\nu}])}{\mathrm{Vol}_{\mathrm{g}}
 \left(\Omega_{N}[\boldsymbol{\nu}]\right)}\,.
\end{equation}
\end{definition}
In the definition (\ref{eq:GM}) under the 
volume of the ``symplectic superspace'' we assume
a result of an average of the symplectic volume of  $\mathrm{Vol}
(\Omega_{N}[\rho, \boldsymbol{\nu}
 ])$ over all possible states 
distributed in accordance with the  measure $\mathrm{d}\mathfrak{m}_{\mathrm{g}}[\rho]\,,$
associated to a certain Riemannian
metric  $\mathrm{g}$ on $\mathfrak{P}_N$:
\begin{equation}
\mathrm{Vol}_{\mathrm{g}}\left(\Omega_{N}[\boldsymbol{\nu}]\right)=\int_{\mathfrak{P}_{N}}
\mathrm{d}\mathfrak{m}_{\mathrm{g}}[\rho]
\, \mathrm{Vol}
(\Omega_{N}[\rho, \boldsymbol{\nu}
 ])
\end{equation}

Below we introduce notions  allowing us to  relate  the definition of the global indicator of  system quantumness $\mathfrak{Q}_{N}[\boldsymbol{\nu}]$
given in terms of the ``symplectic superspace''  with  the corresponding notion formulated 
in terms of the state space $\mathfrak{P}_N\,.$
\begin{definition}
 For an arbitrary point $\omega \in \Omega_N[\rho\, |\, {\boldsymbol{\nu}}]$,  the 
 subspace  $\mathfrak{P}^{(-)}_N
 [\boldsymbol{\nu}\,|\,
 \omega]  \subset \mathfrak{P}_N$ 
 of state space 
 is defined as  
\begin{equation}
 \mathfrak{P}^{(-)}_N
 [\boldsymbol{\nu}\,|\,
 \omega]  =
 \left\{\rho \in \mathfrak{P}_{N} \, | \, \omega \in \Omega_N[\rho\, |\,\, {\boldsymbol{\nu}}] \,, \,  W_{\rho}^{{\boldsymbol{\nu}}}(
 \omega) < 0 \,\right\}\,.  
\end{equation}
\end{definition}

\begin{proposition}
The volume of 
$\mathfrak{P}^{(-)}_N
 [\boldsymbol{\nu}\,|\,
 \omega]$ evaluated with respect to the unitary invariant measure on $\mathfrak{P}_N$ is independent of $\omega\,,$
\begin{equation}
\label{eq:vint}
    \frac{\mathrm{d}}{\mathrm{d}\omega}\,
   \mathrm{Vol}_{g}\left(\mathfrak{P}^{(-)}_N
 [\boldsymbol{\nu}\,|\,
 \omega]\right)= 0\,.
\end{equation}
\end{proposition}
\begin{proof}
Indeed, 
let us write down the volume integral (\ref{eq:vint})
over negativity domain via the 
Heaviside step function $\theta[-W_{\rho}^{{\boldsymbol{\nu}}}(
 \omega) ]$ 
and use SVD decomposition (\ref{eq:SVDSW})
for SW kernel  $\Delta(\omega|\boldsymbol{\nu})$ of the Wigner function 
\begin{equation}
 \label{eq:Volin}
 \mathrm{Vol}_{g}\left(\mathfrak{P}^{(-)}_N
 [\boldsymbol{\nu}\,|\,
 \omega]\right)=\int_{\mathfrak{P}_{N}}\mathrm{d}\mathfrak{m}_{\mathrm{g}}[\rho]\,\theta\left[-\mathrm{tr}(U(\omega) P^{\boldsymbol{\nu}} U(\omega)^{\dagger} \rho)\right]=\int_{\mathfrak{P}_{N}}\mathrm{d}\mathfrak{m}_{\mathrm{g}}[\rho^
 \prime]\,\theta\left[-\mathrm{tr}( P^{\boldsymbol{\nu}}  \rho^{\prime})\right]\,,
\end{equation}
In the last line of (\ref{eq:Volin}) we perform transformation  $\rho^{\prime}= U(\omega)^{\dagger} \rho U(\omega)$. 
Noting that the state space $\mathfrak{P}_{N}$ is $SU(N)$ invariant space endowed  with  the invariant measure we get convinced that  $\mathrm{Vol}\left(\mathfrak{P}^{(-)}_N
 [\boldsymbol{\nu}\,|\,
 \omega]\right)$ is the same for all $\omega \in \Omega_N[\rho\, |\, {\boldsymbol{\nu}}]\,.$
\end{proof}

Based on this observation, afterwards we choose $\omega$ corresponding to the diagonal SW kernels, i.e., $\omega = 0\, $ and simplify notation of the negativity subset,    $\mathfrak{P}_{N}^{(-)}[\boldsymbol{\nu}]$.

Now we are in position to formulate the Proposition which interrelates two ways of interpretation of the global measure of quantumness.
\begin{proposition}\label{prop2}
 The global non-classicality measure  $\mathfrak{Q}_N[\boldsymbol{\nu}]$  can be expressed as the relative volume of the subset  $\mathfrak{P}_{N}^{(-)}[\boldsymbol{\nu}]$ with respect to total volume of state space  $\mathfrak{P}_{N}$:
\begin{equation}
\label{eq:GMQS}
\mathfrak{Q}_{N}[\boldsymbol{\nu}]=
\frac{\mathrm{Vol}_{\mathrm{g}}(\mathfrak{P}_{N}^{(-)}[\boldsymbol{\nu}])}
{\mathrm{Vol}_{\mathrm{g}}(\mathfrak{P}_{N}[\boldsymbol{\nu}])}\,.
\end{equation}
where the volume of  state space is evaluated with respect to the metric $g$
generating the measure
$\mathrm{d}\mathfrak{m}_{\mathrm{g}}[\rho]$
in definition 
(\ref{eq:GM}).
\end{proposition}

\begin{proof}
At first let us note that contribution  to 
(\ref{eq:GMQS})
from components of ``symplectic superspace'' associated to non-generic states (degenerate and non-maximal rank density matrices) is zero owing to zero integration measure of this states.

Hence, the integration effectively projects only to the 
components of ``symplectic superspace'' corresponding to 
the stratum of the generic states, whose isotropy group is conjugated to subgroup  $H=U(1)^N\,.$ Therefore 
the structure of $
\Omega_{N}[\rho, \boldsymbol{\nu}
 ]$ is solely determined by the SW kernel and does not depend on the $\rho$, which means 
 that, 
 \begin{equation}
   \mathrm{Vol}_{g}(\Omega_{N}[\boldsymbol{\nu}])=\mathrm{Vol}_{g}(\mathfrak{P}_{N})\mathrm{Vol}\left(\frac{U(N)}{U(1)^N}\right)\,,  
 \end{equation}
 The same argumentation lead the relation
 \begin{equation}
  \mathrm{Vol}_{g}(\Omega_{N}^{(-)}
  [\boldsymbol{\nu}]=
  \mathrm{Vol}_g{}(\mathfrak{P}^{(-)}_{N}[\boldsymbol{\nu}])\mathrm{Vol}\left(\frac{U(N)}{U(1)^N}\right)\,,    
 \end{equation}
 thus proving the Proposition.
 \end{proof}

\section{Global measure  of quantumness  as  geometric probability}\label{section4}

In this section we outline an interpretation of the above introduced measure of nonclassicality $\mathfrak{Q}_{N}[\boldsymbol{\nu}]$ as a certain geometric probability.
Indeed, according to the representation 
(\ref{eq:GMQS}) the global non-classicality measure  $\mathfrak{Q}_N[\boldsymbol{\nu}]$  can be expressed as the relative volume of the subset $\mathfrak{P}^{(-)}_N[\boldsymbol{\nu}] \in\mathfrak{P}_N\,,$ consisting out of states $\rho$ whose  Wigner functions $W_\rho(\omega\, |\, \boldsymbol{\nu})$ evaluated at some fixed point of  phase space, say $\omega=0$ are negative. 
Therefore, in consent to the Theory of Geometric Probability, this relative volume  can be identified  with the probability of
finding of states with the negative WF among a certain random ensemble of states:
\begin{equation}
\label{eq:GPGM}
\mathcal{P}^{(-)}:= \frac{Number~of~states~with~negative~ WF}{Total~number~of~generated~states }
\end{equation}
Note, that this identification of $\mathfrak{Q}_N[\boldsymbol{\nu}]$ and $\mathcal{P}^{(-)}$ is correct  if 
random states are distributed in ensemble according to the measure
$\mathrm{d}\mathfrak{m}_{\mathrm{g}}[\rho]$
in definition  (\ref{eq:GM}).

Following this identification of  
 $\mathfrak{Q}_{N}[\boldsymbol{\nu}]$ and 
probability $\mathcal{P}^{(-)}$,
we will generate random ensemble of the Hilbert-Schmidt states of $N$\--level quantum system and then construct  the corresponding Wigner functions 
with different kernels and test them on  negativity.

\subsection{Quantumness   of Hilbert-Schmidt states for different SW kernels}\label{section4.1}

There is an elegant method of generation of  random density matrices from the Hilbert-Schmidt ensemble 
. 
Its starting point is is the generation of the so-called {Ginibre ensemble}, i.e., the set of complex matrices whose elements have real and imaginary parts distributed as independent normal random variables. 
Considering a square $N\times N$ complex random matrix $z$ from the Ginibre ensemble, one can construct the density matrix  from the Hilbert-Schmidt
ensemble as 
\begin{equation}
\label{eq:HSG}
\rho_{{}_\mathrm{HS}}=\frac{z^{\dagger}z}{\mbox{tr}(z^{\dagger}z)},
\end{equation}
We have generated such  set of density matrices we evaluate  $\mathfrak{Q}_{N}(\boldsymbol{\nu})$ 
according to (\ref{eq:GPGM}),
for  the next  families of  
SW kernels:
\begin{itemize}
\item Kernels
whose isotropy group is $H=SU(N-1)$, i.e.,   $(N-1)$\-- fold degenerate  eigenvalues $\frac{1+\sqrt{1+N}}{N}$ and one smallest  eigenvalue,  $\frac{1+(1-N) \sqrt{1+N}}{N}$;
\item Kernels whose isotropy is  $H=SU(N-2)\times SU(2)$, i.e.,  $(N-2)$\-- fold degenerate eigenvalues $\frac{2-N-\sqrt{2} \sqrt{(N-2) (N-1) (1+N)}}{(N-2) N}$ and double degenerate eigenvalues $\frac{2-\sqrt{2} \sqrt{(N-2) (N-1) (1+N)}}{2 N}$;  
\item Kernels whose isotropy group is  $H=SU(N-3)\times SU(3)$, , i.e.,  $(N-3)$\-- fold degenerate eigenvalues  
$\frac{3-N-\sqrt{3} \sqrt{3-N-3 N^2+N^3}}{(N-3) N}$ and triple of degenerate eigenvalues, $\frac{3-\sqrt{3} \sqrt{3-N-3 N^2+N^3}}{3 N}$;
\item Random kernels, which almost always are generic.
\end{itemize}
\begin{figure}
    \centering
    \includegraphics{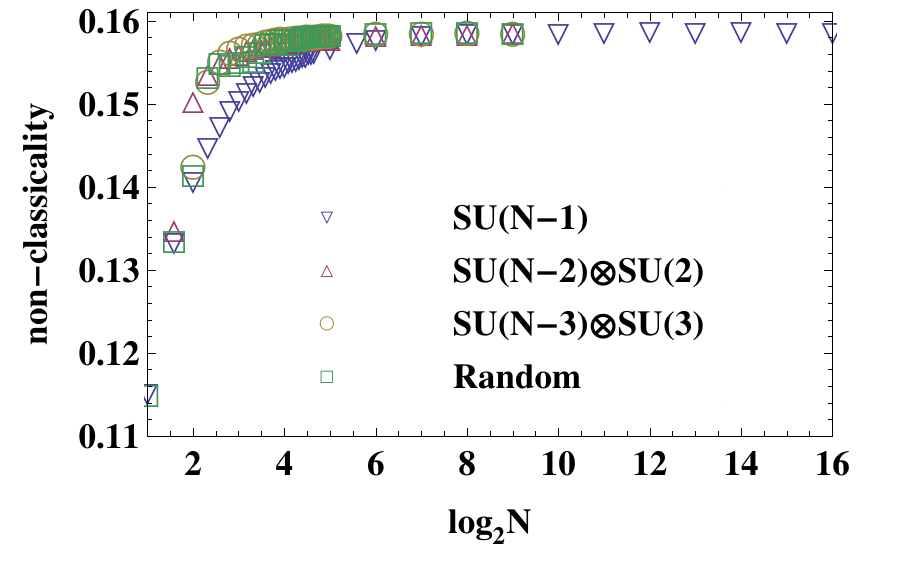}
    \caption{Dependence of $\mathfrak{Q}_{N}[\boldsymbol{\nu}]$ on number of levels  $N$ for different types  SW kernels described in the text. For systems with number of levels greater
    than $N=2^{8}$ the outputs 
    of all but one plots are suppressed, since the difference of values of $\mathfrak{Q}_{N}[\boldsymbol{\nu}]$ for different kernels are within the statistical error.}
    \label{Plot_NonCl}
\end{figure}
In  Fig.~\ref{Plot_NonCl} we have plotted $\mathfrak{Q}_{N}(\boldsymbol{\nu})$ depending on $N$, for different SW kernels.
Approximately $\sim 10^{8}$ matrices have been generated and tested on the Wigner function negativity  for each $N$. 
This plot shows that with growing number of levels the quantumness of system becomes independent of SW kernel and tends to a certain value.  
In the next section we will give argumentation of this universality of $\mathfrak{Q}_{N}(\boldsymbol{\nu})$ for the 
Hilbert-Schmidt ensemble. 

\subsection{The large $N$ limit of global non-classicality }\label{section4.2}

\label{subsection4.1}
\begin{proposition}
In the limit $N\to \infty$ the global non-classicality measure $ \mathfrak{Q}_{N}[\boldsymbol{\nu}]$ of the Hilbert-Schmidt ensemble does not depend on the choice of SW kernel. Furthermore, for the infinite level system the quantumness measure is  
\begin{equation}
       \lim_{N\to\infty}\mathfrak{Q}_{N}(\boldsymbol{\nu})=\mathrm{erfc}\left(\frac{1}{\sqrt{2}}\right)  
\end{equation}

\end{proposition}

\begin{proof} Suppose that SW kernel $\Delta=U P(\boldsymbol{\nu}) U^{\dagger}$ is 
given by $P(\boldsymbol{\nu})=\mbox{diag}||\pi_{1},\pi_{2},\cdots,\,\pi_{N}||$, where the eigenvalues are presented in a decreasing order and that only $m$ of them are non negative. We assume the following notations
\begin{equation}
    Z_{1}=\sum_{i=1}^{m}\pi_{i},\quad Z_{2}=\sum_{i=m+1}^{N}|\pi_{i}|,\quad 
    M_{1}=\sum_{i=1}^{m}\pi_{i}^{2},\quad M_{2}=\sum_{i=m+1}^{N}\pi_{i}^{2}\,.  
\end{equation}
In this notations the equations SW kernel obeys may be rewritten as
\begin{eqnarray}
   && Z_{1}-Z_{2}=1\\\label{sweq1}
   && M_{1}+M_{2}=N.\label{sweq2}
\end{eqnarray}
Remembering that the matrices generated by the procedure $\rho=\frac{z z^{\dagger}}{tr(z z^{\dagger})}$, with normally distributed real and imaginary parts of  $\varrho$, uniformly  cover the set of density matrices with respect to Hilbert-Schmidt measure, we observe that $\mbox{Prob}[\mbox{tr}[P\rho]<0]$ is equal to the probability of event 
\begin{equation}
    \sum_{i=1}^{N}P_{i\,i}\,\rho_{i\,i} < 0\,,
\end{equation}
or alternatively, due to the positivity of ${\mbox{tr}}(z z^{\dagger})$ to the probability of  
\begin{equation}
    \sum_{i=1}^{N}P_{i\,i}\,
    (z z^{\dagger})_{i\,i} < 0\,.
\end{equation}
Equivalently rewriting this event we get
\begin{eqnarray}
\label{inequality1}
\sum_{i=1}^{m}\pi_{i}\,(z z^{\dagger})_{i\,i}
<\sum_{i=m+1}^{N}|\pi_{i}|\,
(z z^{\dagger})_{i\,i}\,.
\end{eqnarray}
Further we denote $\xi_{j}^{(i)}=z_{i,j} z_{i,j}^{*}$,
so that $(z z^{\dagger})_{i\,i}=\sum_{j=1}^{N}\xi_{j}^{(i)}$, and $\beta_{j}^{(i)}=|\pi_{i}|\,\xi_{j}^{(i)}$. In this terms \ref{inequality1}  may be rewritten 
\begin{equation}
    \sum_{i=1}^{m}\sum_{j=1}^{N}\beta_{i}^{(j)}<\sum_{i=m+1}^{N}\sum_{j=1}^{N}\beta_{i}^{(j)}\,.\label{equation2}
\end{equation}
Since $Re(z_{i,j})$ and $Im(z_{i,j})$ are distributed by normal distribution with zero mean and unit variance, then $\xi$'s are distributed with $\chi^{2}_{2}$, distribution while $\beta_{j}^{(i)}$'s are distributed with mean $\mathbb{E}(\beta_{j}^{(i)})=2\,|\pi_{i}|$ and variance $\mbox{var}(\beta_{j}^{(i)})= 4 \pi_{i}^{2}$.
Now according to central limit theorem 
\begin{eqnarray}
   x&=&\frac{ \sum_{i=1}^{m}\sum_{j=1}^{N}\beta_{j}^{(i)}-\sum_{i=1}^{m}\sum_{j=1}^{N}\mathbb{E}\left(\beta_{j}^{(i)}\right)}{\left(\sum_{i=1}^{m}\sum_{j=1}^{N}\mbox{var}\left(\beta_{j}^{(i)}\right)\right)^{{\frac{1}{2}}}}=\\
   &&\frac{ \sum_{i=1}^{m}\sum_{j=1}^{N}\beta_{j}^{(i)}-2 N Z_{1}}{\sqrt{4 N M_{1}}}
\end{eqnarray}
as well as 
\begin{eqnarray}
   y&=&\frac{ \sum_{i=m+1}^{N}\sum_{j=1}^{N}\beta_{j}^{(i)}-\sum_{i=m+1}^{N}\sum_{j=1}^{N}\mathbb{E}\left(\beta_{j}^{(i)}\right)}{\left(\sum_{i=m+1}^{N}\sum_{j=1}^{N}\mbox{var}\left(\beta_{j}^{(i)}\right)\right)^{\frac{1}{2}}}=\\
   &&\frac{ \sum_{i=m+1}^{m}\sum_{j=1}^{N}\beta_{j}^{(i)}-2 N Z_{2}}{\sqrt{4 N M_{2}}}\,,
\end{eqnarray}
are distributed normally with zero mean and unit variance.
Transforming eq.~\ref{equation2} by subtracting from both sides $2 N(Z_{1}+Z_{2})$ and dividing by $\sqrt{4N M_{1} M_{2}}$ we get
\begin{equation}
  \frac{x}{\sqrt{M_{2}}}-\frac{N Z_{2}}{\sqrt{N M_{1} M_{2}}}< \frac{y}{\sqrt{M_{1}}}-\frac{N Z_1}{\sqrt{N M_{1} M_{2} }}\,,
\end{equation}
or taking into account the equations for the Statonovich-Weyl kernel we get

\begin{equation}
    x< y \sqrt{\frac{N-M_1}{M_1}}-\sqrt{\frac{N}{M_1}}.
\end{equation}
Now, let us denote  $t=\sqrt{\frac{N-M_1}{M_1}}$ so that the initial probability is equal to the probability of
\begin{equation}
    x<y\, t-\sqrt{t^{2}+1}.\label{inequality2}
\end{equation}
Since $x$ and $y$ are distributed normally the probability of the event described by equation (\ref{inequality2})
will be
\begin{equation}
    \mathcal{P}(t)=\frac{1}{2 \pi }\int_{-\infty }^{\infty }dy \int_{-\infty }^{y t-\sqrt{t^2+1}}dx \, e^{-\frac{y^2}{2}} e^{-\frac{x^2}{2}}. 
\end{equation}
It may be checked that
\begin{equation}
    \frac{d \mathcal{P}(t)}{d t}=\frac{1}{2 \pi}\int_{-\infty }^{\infty } dy\,\frac{e^{-\frac{y^2}{2}-\frac{1}{2} \left(\sqrt{1+t^2}-t y\right)^2} \left(-t+\sqrt{1+t^2} y\right)}{\sqrt{t^{2}+1}}=0
\end{equation}
Hence, $\mathcal{P}(t)=\mathcal{P}(0)=\frac{1}{2}\, \mbox{erfc}\left(\frac{1}{\sqrt{2}}\right)=0.158655$, where $\mbox{erfc}$ is the complimentary error function. Which proves the preposition.
\end{proof}
\section{Conclusions}\label{section5}
Summarizing, in this work we have introduced a global measure of non-classicality of the state space  of $N$-level quantum systems. By computer simulations the measure was computed for several SW kernel families  depending on the number of levels. It was proven that for large $N$ the measure of non-classicality does not depend on the choice of the SW kernel, thus it is an invariant over the moduli space of SW kernels. 

However, it must be noted, that it is  unreasonable to suppose that a single measure might capture all the non-classical aspects of quantum states, let alone of the whole state space. Hence, there must be different ways of defining non-classicality measures underlining this or that features of quantum behaviour (see for comparison \cite{AVA2}). 
\section{Acknowledgments}
The publication has been prepared with the support of the “RUDN University Program 5-100” (recipient V.A.). 

\end{document}